\documentclass[12pt]{article}  
\usepackage{geometry}
\usepackage[utf8]{inputenc}
\usepackage[T2A]{fontenc}
\usepackage{graphicx}
\usepackage{caption}      
\usepackage{nameref}      
\usepackage[title]{appendix}
\usepackage{titling}
\usepackage{hyperref}
\usepackage{tabularx}                                               
\usepackage{float}                                                  
\usepackage{pifont}                                                 
\usepackage{xcolor}                                                                                                                      
\restylefloat{table}                                                                                                                     
\usepackage{amsmath, amsfonts, amssymb, amsthm}            
\usepackage{listings}  
\usepackage{mathtools}                                             
\usepackage{pdfpages}                           

\usepackage{enumitem}
\usepackage{listings}
\usepackage{subcaption}
\usepackage{makecell}
\theoremstyle{definition}      
\newtheorem{definition}{Definition}[section]
\newtheorem{defn}{Definition}[section]
\newtheorem*{defn-non}{Definition}

\newtheorem{theorem}{Theorem}
\newtheorem*{theorem-non}{Theorem}

\newtheorem*{problem-non}{Problem}

\providecommand{\keywords}[1]{
	\begin{flushleft}
	\small\textbf{Keywords: } #1
	\end{flushleft}
	}
\providecommand{\affili}[1]{
	\begin{center} 
	\small #1
	\end{center}
	\vspace{.1cm}}
	
\title{SafeComp: Protocol For Certifying \\Cloud Computations Integrity}
\author{Evgeny Shishkin$^1$ \\
 Evgeny Kislitsyn$^{1,2}$}

\date{}

\begin{document}
\renewcommand{\abstractname}{Abstract}
\renewcommand\qedsymbol{$\blacksquare$}
\renewcommand*{\proofname}{Proof}
\large
\maketitle
\par\vspace{-60pt}

\affili{
	$^1$InfoTeCS, Advanced Research Department, Russia\\
	$^2$Moscow State University, Department of Mathematics and Mechanics, Russia\\
	\texttt{evgeny.shishkin@infotecs.ru\\
		    evgeny.kislitsyn@infotecs.ru}
}
\abstract{
  We define a problem of certifying computation integrity
  performed by some remote party we do not necessarily trust.
  We present a multi-party interactive protocol called SafeComp that solves 
  this problem under specified constraints.
  Comparing to the nearest related work, our protocol reduces a proof construction complexity from 
  $O(n \log{n})$ to $O(n)$, turning a communication complexity to exactly one round using a certificate of 
  a comparable length.
}
\section{Introduction}
Suppose, a user needs to compute some complex heavyweight function $C(x)$. 
It might be necessary to possess a non-trivial computational resource to do that, for example, a cluster of computational
nodes. Chances are, the user does not have it.
To compute this function, the user asks some cloud service - 
a provider - to compute it for him.
The user supplies the provider with the function $C(x)$ and some
initial data $d$, asking to compute $C(d)$,
and some time later the provider gives the
result $r$. How can the user be sure that the provided result $r$ is correct?

There are many cases when a computation result can be verified
easily. Consider a sorting function that maps an arbitrary list of integers to a sorted list of the same integers.
We can check that the output list is sorted using a simple
one-way scanning function that ensures the sameness of input/
output list elements and their expected order. 
 
On the contrary, there are cases when it is hard or maybe
even impossible to verify a
computation result without redoing the whole computation
from scratch.
Take a problem of checking that \textit{there are no} 
Hamiltonian cycles within a given graph.
\footnote{
Hamiltonian cycle is a cycle in a graph that traverse
each vertex exactly once.}
If the provider gives you a founded Hamiltonian cycle, it is
easy to check its correctness. 
But what if the provider tells you that there
are \textit{no such 
cycles}: how would you ensure that it tells you the truth?

Consider yet another problem: a model-checking of computer programs - the well-known computationally hard problem. 
Suppose, the user supplies the provider with a program and a 
list of properties it must respect, and asks if this really
holds. Some time later, the provider tells the user that a checking
procedure did not find any deviations from the specified
properties.
How would you check that: 1) the result is correct 
2) the result was not simply guessed (i.e. 
the computation was not halted until the
result was obtained)?

In modern practice, a user has to trust in provider's integrity, rely on its market authority. A computation
result is signed with a public key of the 
provider linking that trust to the obtained result.  
In this article, we present SafeComp: a multi-party interactive protocol
that solves the described problem technically, removing the
necessity to trust the authority of a cloud provider.

We still rely on a so-called \textit{trusted weak computing device}, say, a smart-contract residing on a public blockchain.
This device plays a role of autonomous transparent arbiter
that is able to decide who is right and who is wrong in case of a dispute regarding published computation results, and do that
performing only a tiny fraction of the whole computation.
 
 Comparing to the nearest related work - TrueBit protocol - our 
 protocol reduces proof construction complexity from $O(n\log{n})$ to $O(n)$ and turning communication complexity from $O(\log{n})$ to exactly one round. 
This result is achieved by using a different proof construction and verification procedure, while the proof
length stays effectively the same: $O(n)$. 

 The article is structured in the following way: in Section \ref{problem} we define the problem that our protocol is aiming to solve.
 In Section \ref{solution_simple}, we give a high-level overview of SafeComp protocol.
 Section \ref{fixpoint} describes a way of presenting user
 computation function in a form that is appropriate to be
 used in the protocol.
 Section \ref{solution} contains a formal description of SafeComp protocol.
 Section \ref{security} is devoted to a security analysis of
 some aspects of the protocol under clearly stated threat model.
 In Section \ref{attacks}, we discuss some of the known attack vectors and features of the protocol.  
 Our first experimental results are discussed in Section \ref{experiments}. 
 Section \ref{related} contains the overview of related works
 in the field, and comparison of SafeComp with other 
 known protocols.
 
\section{Problem Statement}\label{problem}
We are looking for an efficient verification 
procedure that allows us to check computation integrity giving that the computation was 
performed by some not necessarily fully trusted computing provider (or, simply, provider). 
Along with a computation result, the provider
constructs a certificate whose size reasonably correlates with computation complexity.
Using the certificate, we should be able to  unambiguously check correctness of the performed computation, and do that much faster than a full
recomputation.

\begin{defn}(Computation Certification Problem)
\label{def:cert}\\
  For an arbitrary function 
  $C: \mathbb{N} \rightarrow Result$
  computable at a point $d \in \mathbb{N}$,
  define computable functions:
  $$ proof(C, d) \rightarrow Result \times Proofs $$ 
  $$ verify(C, d, r, prf) \rightarrow \{ True, False \} $$
 where $r \in Result$ is a 
 result of computation under question, 
 that is
 $r \stackrel{?}{=} C(d)$,\\ 
 $prf \in Proofs$ - a proof of computation integrity,
 also called \textit{certificate}. 
 \smallbreak
 The following must hold:
 \begin{enumerate}
 \item $verify(C, d, r, prf) = True \iff C(d) = r$ 
 \label{item:verify}
 \item Computation complexity of $proof(C, d)$ 
 dependents on complexity of $C$ linearly. 
 \item Computation complexity of $verify$ is
 essentially smaller than complexity of $proof$ for any given function $C$, that is
 $$\lim_{n \rightarrow \infty}\frac{v_C(n)}{e_C(n)} = 0$$
 where $v_C(n), e_C(n)$ - asymptotic estimates of
 computation complexity of functions $verify$ and $proof$
 for the function $C$ with input size $n$.
 
 \item As complexity $e_C(n)$ increases, the size of $prf$ grows linearly, that is
 $$ \exists k \forall n \, . \, \vert prf(n) \vert \le k \cdot e_C(n) $$
 
 \end{enumerate}
\end{defn}

\subsubsection*{Probabilistic 
Certificate Verification}
In our case, the requirement \ref{item:verify} 
turns out to be overly rigorous. We will use a weaker requirement instead.

\begin{defn}(Probabilistic Certificate Verification Criteria)\\
Consider the following events:
$$E_0(\lambda) : 
proof_{\lambda}(C, d) = \langle r, prf_\lambda \rangle $$ 
$$E_1(\lambda) : verify_{\lambda}(C, d, r, prf_\lambda) = True $$
$$E_2 : C(d) = r $$
If functions $proof_\lambda$, $verify_\lambda$ satisfy:
$$\lim_{\lambda \rightarrow \infty} 
\Pr [ E_0(\lambda) \cdot E_1(\lambda) \cdot E_2 ] = 1$$ 
where $\lambda$ is a parameter that can affect computational complexity of functions and the size of a certificate,
then we say that algorithms described above meet probabilistic certificate verification criteria.

\end{defn}

In other words, we are depicting not deterministic, but
probabilistic method of solving the defined problem,
where some special parameter regulates the probability 
of the desired outcome.

\subsubsection*{Interactivity}
Some approaches aiming to solve the defined problem require a so-called 
\textit{interactivity}, that is, a verification
 procedure is performed in several rounds of message
 passing between a computing provider and a user.
An upper bound on the number of message-passing rounds 
is called \textit{communication complexity} of a protocol.
An interactive approach for solving the depicted problem 
is also called \textit{protocol for certifying outsourced
compu\-tations}, or, cloud computations.

\subsubsection*{Trusted Weak Computing Device}

Some certification protocols rely on a presence 
of some trusted party which is able to perform
simple computations in a verifiable way, but 
definitely not able to perform the whole computation $C(d)$.

For instance, such device may not be able to compute
a factorial of some big number, but still it 
could perform
an addition of any two natural numbers,
providing some kind of proof of the result. 
If, say, we divide the whole computation of a factorial function
into smaller operations of addition, we could carry out 
the entire computation in a verifiable manner.

\begin{definition}(Trusted Weak Computing Device) \\
A device that is capable of certifying computations \ref{def:cert} for computable functions $f$ such that
computational complexity of $f(n)$ grows linearly with
the size of input, i.e. $ f(n) \in O(n) $, is called \textit{trusted weak computing device}. 
\end{definition}

\noindent Secure enclaves \cite{secure_enclave} and blockchain smart-contracts \cite{buterin2018ethereum} are examples of such devices.

\subsubsection*{Multi-party Protocol}
A classical interpretation of the computation certification problem assumes that a computing provider proves correctness of an obtained result to a user in some way, and the user
examines this proof. An interaction between two parties
happens directly.

 Our model of interaction is different: 
 a user and the set of computing providers (or, simply, providers) interact with each other using a Trusted Weak
 Computing Device (TWCD) as a trusted verifiable
  intermediary. 
 Several providers can run a
 computation task in parallel on their own computing resources. A computation result and 
 a proof are then published within TWCD by some provider who has found a solution. Other providers are able to re-check
 the result. If a disagreement arises, other providers
 would be able to convince other parties that the result 
 written in TWCD is faulty. Otherwise,
 after a period of time, the published result is 
 considered to be correct and a user receives the solution
 to the published computation task.
 Such interaction model is called \textit{multi-party protocol}.

\subsubsection*{Problem Statement}
In this article, we present a multi-party protocol for
 certifying cloud computations with a probabilistic
 certificate verification criteria and a reliance on
 a trusted weak computing device that is able to
 authenticate its users.

For ease of presentation, \textit{we consider a public blockchain smart-contract to play a role of a trusted weak computing device}. But one can imagine other devices being used.

\section{Protocol Overview}
\label{solution_simple}
In this section, we give a high-level overview
of the protocol. Here we presume that a 
user would
like to compute the function $C(x)$ in the 
point $d$, but do not have an ability to
do the whole computation.

\begin{enumerate}
\item 
A user transforms his computable function 
$C$ into
another function $f$ such that applying $f$
to $d$ $n$ times, the result converges to $C(d)$. 

 \item 
 User publishes the function $f$ and the point $d$
 in a smart contract. We assume that the
 smart contract is accessible for multiple
 computing providers.
 
 \item 
 Once a provider receives the computation task $f$ and the point $d$, he starts performing the
 computation $r = f(f(...f(d)))$ until the
 result converges to a fixpoint.
 At the end of each iteration $f(x)$, he also computes a hash $c_i \coloneqq H(x \circ c_{i-1})$.
 These $c_i$ values form a sequence $ cert \coloneqq \langle c_1, ..., c_n \rangle $
 This sequence is called a \textit{certificate sequence} (or, simply, a certificate).

 The  provider computes a certificate fingerprint
 $ hc \coloneqq H(H(cert))$ 
 and a \textit{certificate projection} \label{cert_proj}
 $ cp \coloneqq \langle \pi(c_1), \pi(c_2), ..., \pi(c_n) \rangle $.
 The main idea behind introducing the fingerprint 
 is to give other providers ability to check their own
 certificates against the computed one without publishing
 a certificate sequence in pure. 

 A certificate projection is needed to let other providers find 
 the very first divergent step of a computation in 
 case of a fingerprint disagreement, and do that
 without disclosing the certificate sequence in
 pure. 
 
 A certificate sequence plays a role of 
 a secret value used in a later stage of the 
 protocol to establish a list of all fair
 providers that performed the verification of
 a published result and a certificate, this 
 is why we do not want to publish it in pure. 
 
 \item 
 After finding a solution, one of the  
 providers - we call him \textit{the solver} - publishes the solution $r$,
 the certificate fingerprint $hc$ and the projection $cp$ in a smart-contract. 
  
 \item 
 Other providers that performed the computation, but found the solution
 later than the solver, are able 
 to check the published solution and the 
 computation work integrity by comparing
 certificate fingerprints. This group
 of providers is called \textit{auditors}.
  \item 
  In case of a disagreement, an auditor
  sends the number of a certificate part $i$ 
  from which the divergence starts, 
  the partial value $r_{i-1}$ and 
  certificate parts $c_{i-1}$, $c_i$. 
  \item 
 Using the provided values 
 $i, r_{i-1}, c_{i-1}$, $c_i$ 
 the smart-contract arbiter ensures
 that the published result or 
 computation certificate is indeed
 wrong by performing only a single computation step (details are in Section \ref{solution}).  
  
  \item 
   After some period of time,
   if no proofs of computation disintegrity
   has been provided, the published 
   solution $r$ is considered correct. 
    
  \item 
  Each auditor sends the value $ H(H(cert) \circ id)$ into the smart-contract. Here $id$ is a unique identifier of an auditor. This
  value plays a role of a proof that the computation work has been done.
  \item 
  Finally, the solver sends the value $s$ such that 
  $hc = H(s)$ into the smart-contract.  
  The list of fair auditors is then evaluated.
\end{enumerate}

\section{Iterative Computation}
\label{fixpoint}

\begin{figure}[ht]
\centering
\begin{subfigure}[t]{.35\textwidth}
\begin{lstlisting}[language=erlang,basicstyle=\footnotesize]
fact(0) ->
    1;
fact(N) when N > 0 ->
    N * fact(N-1).
\end{lstlisting}
\caption{Recursive implementation}
\label{fig:fact}
\end{subfigure}\,\,%
\begin{subfigure}[t]{.4\textwidth}
\begin{lstlisting}[language=erlang,basicstyle=\footnotesize]
factFP({0, Acc}) ->
    {0, Acc};
factFP({N, Acc}) when N > 0 ->
    {N - 1, Acc * N}.
\end{lstlisting}
\caption{Iterative implementation}
\label{fig:factfp}
\end{subfigure}
\caption{Recursive vs. Iterative factorial implementation (Erlang)}
\label{fig:fpexample}
\end{figure}

The protocol requires the computation to be 
prepared in an iterative form. 
An example is shown on Fig.\ref{fig:fpexample}:
factorial function is implemented in two
different ways: the standard recursive form
\ref{fig:fact} and iterative form \ref{fig:factfp}.

If you consecutively compute the function 
$factFP$, starting at the point $\{N, 1\}$ and
passing the evaluated result as a next point iteratively,
then, after $N$ iterations, the 
function will converge to the value  $\{0, N!\}$.

The main difference between two implementations
is that the first one (Fig.\ref{fig:fact}) computes
the result at once, while the second (Fig.\ref{fig:factfp}) computes only a part of the whole computation and returns that partial result.
The computation can be continued by passing the
partial result into the function again, and so on,
 until a fixpoint - the solution - is found. 
 
One can start feeling doubt whether it is always
possible to transform an arbitrary computable
function $C(x)$ into such form.
The next theorem shows that not only it is always possible, but can be done in several ways.

\begin{theorem} \label{thm:fp}
For all functions $C: \mathbb{N} \rightarrow \mathbb{N}$ computable at some point $d$ there exist computable functions
$$inj: \mathbb{N} \rightarrow T, 
proj : T \rightarrow \mathbb{N}, F: T \rightarrow T$$ such that
$$ proj (\textbf{fix} \,\, F \, inj(d)) = C(d) $$
where $\textbf{fix} : (T \rightarrow T) \times \mathbb{N} \rightarrow T$ - an operator that calculates fixed point of a function starting at a given point. That is,
$$ \textbf{fix}\,F \, p = F(\textbf{fix} \, F \, p) $$
$T$ - a finite set.
\end{theorem}
\begin{proof}
The proof is moved into Appendix \ref{appendix:proofs}.
\end{proof}

\section{Protocol Specification}\label{solution}
SafeComp protocol rely on availability of a trusted 
weak computing device. We take a smart-contract as a 
practically viable form of such device. Parties - a user and providers - communicate with each other by sending
\textit{transactions} into the smart-contract. 

Lets us define a smart-contract state to be a set of 
its internal variables of some type. Any transaction
 into a smart-contract can potentially change
its state. Within our specification effort, we use 
a concept of a current state and a next state of a 
variable denoting the variable value before and after 
transaction processing respectively.
\bigbreak

\noindent\textbf{Notation}

$\mathbb{N}_p = \{ 0 \, ... \, 2^p - 1 \}$ - 
a set of natural numbers with a given upper bound

$ X, X' $ - a value of variable $X$ before and after transaction processing

\bigbreak

\noindent\textbf{Variables Denotation}

$Id \in 2^\mathbb{N} \cup 
\{ 0 \}$ - 
a set of user identifiers. We use $0$ to denote 
\textit{an undefined} identifier.

$r \in \mathbb{N}$ - a computation result 

$s \in \mathbb{N}$ - a certificate fingerprint 

$solver \in Id$ - 
an identifier of a provider who published a solution 
for the user's task

$V \in 2^{Id}$ - a set of provider identifiers
that managed to prove their verification work.

$L \in 2^{Id}$ - 
a set of participant identifiers
that tried to compromise the protocol in any way 

$P \subseteq Id \times \mathbb{N} $ - 
A set of pairs - identifier $\times$ proof - sent by parties
pretending to be considered as verifiers.

\bigbreak

\noindent\textbf{Input}: User's task $f : \mathbb{N} \rightarrow \mathbb{N}$ and an input data
$d \in \mathbb{N}$.

\noindent\textbf{Output}: $\langle r, s, solver, V, L \rangle$
\bigbreak
\noindent\textbf{Preliminary Setup}
\begin{enumerate}
\item The User defines the function $f$ and the initial point $d$. We discussed requirements for $C$ and $f$ and their relation earlier.

\item Parties agree on some cryptographically secure hash function
$H : \mathbb{N} \rightarrow \mathbb{N}_q$, where 
parameter $q$ is taken according to recommendations for the
hash-function.  

\item The User constructs a computable function
$$ F(\langle x, c \rangle) \coloneqq 
\text{\textbf{IF} x \textbf{==} f(x) \textbf{THEN}} \,
\langle x, c \rangle\,
\text{\textbf{ELSE}} \, \langle f(x), H(x \circ c) \rangle
$$

Here $\circ : \mathbb{N} \times \mathbb{N} 
\rightarrow  \mathbb{N}$ - some non-constant binary operation.
Computation must start at the point $\langle d, H(d) \rangle$

\item 
Parties agree on a function family
$\pi_k : \mathbb{N}^k \rightarrow \mathbb{N}_p^k$, for $k \ge 1$,
such that
$$\pi_1 : \mathbb{N} \rightarrow \mathbb{N}_p  - 
\text{some arbitrary hash-function (not required to be equal to $H(x)$)} $$
$$\pi_k(\langle c_1, ..., c_k\rangle) \coloneqq 
	\langle \pi_1(c_1), ..., \pi_1(c_k) \rangle \text{,   for $k > 1$} $$
Henceforth, we omit $k$ index because its value is clear
from a context: it is a length of a tuple argument.

\item A smart-contract must be published on a blockchain. The smart-contract should give its users ability to:
 \begin{itemize}
 \item
 Publish a computational task $F$, initial data point $d$;
 the user request status is set to ``published''.  
 For every new user request, the smart-contract allocates
 a new set of state variables: 
  $\langle r, s, solver, V, L, P \rangle$, initially empty. 
 \item 
  For any published user request, receive a user request status and a computational task 
  $\langle F, d \rangle$.
 If a request status is set to ``completed'', then also
 receive a certificate projection $cp$ and a solution $r$.
 If a request status is set to ``verified'', then also 
 receive values $\langle r, s, solver, V, L \rangle$
 
 \item 
 For requests with status ``published'', 
 provide the solution $ r_n, c_n, cp, hc $ 
  , such that
 $$ F(\langle r_n, c_n \rangle) = \langle r_n, c_n \rangle$$
 $$ cp \coloneqq \pi(\langle c_1 , ..., c_n\rangle) $$
 $$ hc \coloneqq H(H( \langle c_1, ..., c_n \rangle )) $$
  \bigbreak
  
Let $u$ be an identifier of a provider that sent a solution - the solver.
The smart-contract checks that the pair
$ \langle r_n, c_n \rangle $ 
is indeed a fixpoint of $F$, in this case
the user request status is changed to ``completed'' and
$$solver' = u \land r' = r_n$$

If the point $\langle r_n, c_n \rangle$ 
is not a fixpoint of $F$, the solution is declined, the
user request status does not change. In this case, 
 $$L' = L \cup \{ u \}$$
 
\item 
For user requests with status ``completed'', refute the
published solution by providing the refutation data 
$i$, $c_{i-1}$, $c_i$, $r_{i-1}$, such that
\begin{align}
\pi(c_ {i-1}) = cp[i-1] \, \land \,
&F(\langle r_{i-1}, c_{i-1} \rangle) = \langle r_i, c_i \rangle \land
\pi(c_i) = cp[i] \label{refute_check1} \\
&c_{i+1} \coloneqq H(r_i \circ c_i) \label{refute_check2}  \\
&\pi(c_{i+1}) \neq cp[i+1] \label{refute_check3}
\end{align}
 
Here $cp[i]$ denotes an element of a tuple $cp$, residing 
on the position $i$.

  Let $u$ be an identifier of a provider that has sent
  a refutation. If the smart-contract establishes the
  fact of refutation correctness, the user 
  request status is changed back to ``published''.
  In this case, 
	 $$L' = L \cup \{ solver \} \land
	 V' = V \cup \{ u \} \land 
 	solver' = 0 $$
 \bigbreak
 
 \item
  For user requests with status ``completed'', provide 
  a proof of computation $prf$, such that:
  $$prf := H( H( \langle c_1, ..., c_n \rangle) \circ id)$$
  where $id$ - a unique identifier of a provider.
  In this case,
   $$P' = P \cup \{ (id, prf) \}$$
  
 \item 
 For user requests with status ``completed'', after verification period $T$ has elapsed, 
 provide a secret $s$
 from a user with identifier $solver$, such that:
 $$hc = H(s)$$  
 If the above does not hold, then the status is changed 
 to ``published'' and
 $$ L' = L \cup \{ solver \} $$
 $$ solver' = 0 $$
 Otherwise, the status is changed to ``verified''. In this case,
 $$ V' = V \cup \{x : (x, p) \in P \land H(hc \circ x) = p \} $$
 $$ L' = L \cup \{x : (x, p) \in P \land H(hc \circ x) \neq p \} $$
 \end{itemize}
\end{enumerate} 
 \noindent\textbf{Protocol Steps}
 \begin{enumerate}
\item 
A user publishes a computation request, providing $F$, $d$
\item 
 Potential computing providers read computation requests with the status ``published'' from the smart-contract and possibly
 start to compute $F$ from the point
 $\langle d, H(d) \rangle$

 \item A provider publishes a computed solution 
 $r_n, c_n, cp, hc$ in the smart-contract. We call this
 provider \textit{the solver}.
 
 \item 
 Other providers that were also performing computation,
 verify the published result by comparing the value $cp$
 to their own. 
 \begin{itemize}
 \item 
  If values are different, then send a refutation data
  into the smart-contract.
 \item 
 If values are the same, after verification period, send
 the proof of computation into the smart-contract.
 \end{itemize}
 \item 
 Last of all, the solver sends $s$, the hash of a certificate, into the smart-contract,
 making it possible to find out all fair providers.
 
 \begin{itemize}
 	\item 
 	In case of successful secret validation, the smart-contract constructs an output:
 		$\langle r, s, solver, V, L \rangle$
 	\item 
 	Otherwise, the user request status is changed back
 	to ``published'' and the protocol goes back to step 2.
 	\end{itemize}
\end{enumerate}

\section{Protocol Security Analysis}\label{security}
Any participant of the protocol trying to break its
functional properties is called \textit{intruder}.
We measure protocol security degree as a probability
value of any undesirable event - that is when  
functional properties of the protocol gets broken - within a given intruder model.

\begin{defn}(Intruder model $M_1$)
\begin{itemize}
	\item Finding a pre-image of 
	a chosen hash-function $H(x)$ is considered a computa\-tionally
	difficult task for an intruder.	
	\item Intruder's knowledge regarding user's 
	function $f(x)$ implementation details 
	is consi\-dered to be the same as of a user itself.
	\end{itemize}
\end{defn}

\begin{defn}(Protocol Security Threats)
\begin{itemize}
	\item Event $E_1$ : Refutation of a correct solution.
	\item Event $E_2$ : False proof of  
	a solution verification.
\end{itemize} 
\end{defn}

\begin{theorem}
	A refutation of a correct published solution (Event $E_1$)
	is a computationally difficult task for an intruder from $M_1$.
\end{theorem} 
\begin{proof}
	To refute a published solution, the protocol requires providing such $i$, $c_{i-1}$, $c_i$, $r_{i-1}$
	that
	\begin{align}
		\pi(c_{i-1}) &= cp[i-1] \land \pi(c_i) = cp[i] \land H(r_{i-1} \circ c_{i-1}) = c_i 
		\label{eq:fst}\\
		&F(\langle r_i, c_i \rangle) = (r_{i+1}, c_{i+1}) \land \pi(c_{i+1}) \neq cp[i+1] \label{eq:sec}
	\end{align}
	Suppose, $F$ is a total function $\mathbb{N} \times \mathbb{N}_q \rightarrow \mathbb{N}$ - in practice, this is almost always not the case, 
	but let us take this assumption to simplify 
	work for an intruder in our analysis. 
	In this case, it is easy to satisfy \ref{eq:sec}.
The requirement could be fulfilled for almost any
point $\langle r, c \rangle$. Therefore, the main
burden falls on the predicate \ref{eq:fst}.
	
	Suppose that hash-functions have the following signatures:
	$H: \mathbb{N} \rightarrow \mathbb{N}_q$, 
	$\pi:\mathbb{N}_q \rightarrow \mathbb{N}_p$.
We also require $H$ to be cryptographically secure.
Then, a problem of finding values satisfying 
the predicate \ref{eq:fst} is 
reduced to the known problem of finding a pre-image by a given image of 
a cryptographically secure hash-function, and is considered to be hard 
for an intruder in $M_1$ with appropriately 
chosen $p,q$ values.
\end{proof}

\begin{theorem}
Submitting of false proof of verification for a
published solution (event $E_2$) is 
computationally difficult task for any intruder 
from $M_1$.
\end{theorem}
\begin{proof}
	To prove a published solution verification work,
	the protocol requires a participant to provide 
	the following value:
	$$ prf = H(s \circ id)$$
	Here, $id$ is a unique identifier of a participant.
	Every unique $id$ is allowed to publish 
	only a single $prf$ as a proof.
	Additionally, we consider an intruder to know the following values:
	\begin{align}
		&hc = H(s) \label{eq:hc}\\
		prf&_i = H(s\, \circ \, id_i) \label{eq:prf}
		, \text{  for } i = 1,2,3,...
	\end{align}
	Identifiers $id_i$ are also known to the intruder.
	Here, $s = H( \langle c_1, c_2, ... c_n \rangle )$  is a fingerprint of the certificate. 
	A probability of finding the desired 
	pre-image of $s$ is:
	$$ P_s = \frac{1}{2^{q \cdot n}}, 
	\text{где  }, q \ge 64, n \gg 10^3 $$
Here, values for $q$ - a hash-function image size, and $n$ - a number of computation steps, 
are specified according to practical considerations.
So, in this case, an event $E_2$ can be considered to be
highly improbable. 
	
An intruder could make an attempt to find a collision in
\ref{eq:hc}: find $s'$ such that $H(s') = hc$. But, in
this case, with carefully selected operation $\circ$,
a probability of finding satisfying values 
for \ref{eq:prf} with the founded collision $s'$ 
is also insignificant.
\end{proof}

\subsubsection*{Economic Incentives}
Earlier, we have discussed some aspects of the protocol security relying on an assumption that the 
protocol's cryptographic components is secure enough.
It is possible to
enhance reliability of the protocol even further by
exploring the immanent feature of smart-contracts:
\textit{economic incentives} for its users - a powerful
tool for stimulating rational parties strictly follow protocol rules. 

In this paradigm, a stability of a
protocol depends not only on security of cryptographic mechanisms, but also
on an incentive scheme consisting of 
penalties and premiums for participants. 
We are not going to discuss this in deep, but 
give a high level overview of such scheme in form of 
an extension to protocol steps.

\begin{enumerate}
\item  
User publishes a computation request
in a smart-contract, supplying the
request with $F$, $d$ and a cryptocurrency deposit $D_r$.

\item 
Providers read the request and start seeking a 
solution.
\item 
A provider who computed the solution
first, publish his result in the 
smart-contract, supplying
the transaction with a cryptocurrency
deposit $D_s$

\item 
Other providers verify the published
result.
  \begin{itemize}
  \item 
  If a divergence is found, one publishes a refutation in the smart-contract, supplying the transaction with a cryptocurrency deposit $D_p$ 
  \item 
  If no divergence is found, one sends a proof of computation in the smart-contract, supplying the transaction
  with a cryptocurrency deposit $D_w$  
  \end{itemize}
\end{enumerate}

At the end, when a solution is found
and verified, a smart-contract has
an amount of cryptocurrency on its account equal to
$s = D_r + D_s + m \cdot D_p + n \cdot D_w $
where $n$ - number of participants
successfully proved the computational work, 
$m$ - number of refuted solutions.
The protocol output is a tuple:
$\langle r, s, solver, V, L \rangle$, 
so the smart-contract is able to redistribute funds $s$ between
a $solver$ and fair auditors $V$,
penalizing bad actors from $L$. 

\section{Protocol Features}\label{attacks}
Some aspects
of SafeComp protocol may cause difficulties in a 
practical implementation, or even present attack vectors. Here, we discuss some of
those issues and ways to overcome them.

\bigbreak
\noindent\textbf{
Transforming computation into iterative form.}
In theorem \ref{thm:fp}, it was shown
that a computation could be always presented in an iterative form and
there are multiple ways to do the transformation.

For example, in the work 
\cite{truebit} the following approach
is used: a user's program is compiled
from high-level language into a
byte-code of some virtual machine.
Then, $n$ elementary steps of interpreter is taken as a single computation step. 
A state of the interpreter - registers and memory - 
is then taken as a computation state. 
The next step is executed starting from 
that interpreter state, and so on, until
the solution is computed. 

We use a different approach: the user's
program is written in a so-called
\textit{Continuation Passing Style (CPS)}
\cite{reynolds1993discoveries}, but
instead of a tail call, the argument
together with a tail call function tag 
is returned as an intermediary result.
Such representation achieves the same
technical result while may sometimes
lead to more compact state representations and does not rely
on a special interpreter.
A program's source code can be
automatically transformed into 
CPS-style form \cite{appel1989continuation},
but we did not investigate this 
question in deep.
\bigbreak
\noindent
\textbf{Divergent function $F$.}
Suppose someone places a divergent
computation task in a smart-contract.
A computational provider will spend
its computing resource, but will never
find a solution: an obvious loss for
a provider. 

There are some ways to overcome this
difficulty:
 \begin{itemize}
 \item 
 Implement a computational algorithm in a programming
 language that \textit{guarantees}
 termination. 
 For example, any \textit{primitive recursive language} gives such guarantee by construction, and
 is expressive enough to develop 
 nearly all practically interesting
 algorithms \cite{vereshagin}.
 
 In this case, a user sends not a byte-code, but a 
 program source code $P$, written in one of those
 languages. 
 A provider then compiles $P$ into $f$.
 If compilation is successful, the provider
 is guaranteed that the computation
 is finite and can be taken into work. 

\item 
 Extend the protocol to support 
 certification of partial computation
 results, i.e. results that have not
 converged to a fixpoint, but is 
 correct from a certificate point of
 view.
 In this case, the longest result respecting
 the certificate integrity check is considered to 
 be a correct solution. 
 \end{itemize}
\bigbreak
\noindent
\textbf{Bytecode size of $F$.}
It might be the case that the size
of a function $F$ or a point $d$ is
greater than a maximum size of data unit of chosen blockchain platform. 

Unfortunately, it is a fundamental limitation of the protocol. The 
following inequation must hold:
$|F| + |d| \le T_{max}$, where $T_{max}$ - 
maximum size of data unit for a 
chosen TWCD.
In case of smart-contracts, if this
condition does not hold, such transaction will be declined.
\bigbreak
\noindent\textbf{Size of refutation.}
If a dispute against a published solution arises, a provider must
send a refutation data
$c_{i-1}, c_i, r_{i-1}$ into the
smart-contract. 

The point $r_{i-1}$ 
represents an intermediary computing
state.
It might be the case that the size of
$r_{i-1}$ is too big to be sent into
the smart-contract. 
Unfortunately, it is a fundamental limitation of the protocol. The 
following must hold: 
$ \forall r, |f(r)| \le T_{max}$,
where $T_{max}$ - maximum size of data unit for a 
chosen TWCD.
\bigbreak
\noindent\textbf{
Certificate projection size.}
It might be the case that a certificate
projection is too big to be
placed directly into a smart-contract
or other chosen TWCD. 
At least in case of smart-contracts, it is possible to use
an external immutable storage to store the projection there. For example, 
IPFS distributed file system may be used
\cite{benet2014ipfs}. 
A smart-contract could use an oracle
technology to access the corresponding data \cite{kochovski2019trust}. 

\section{Experimental Evaluation}\label{experiments}
In order to evaluate SafeComp protocol,
we have implemented its logic and one user 
computation task.

Due to time constraints, we have
chosen not implement the protocol logic
on a real blockchain.  
We implemented the protocol using Erlang programming language
and its BEAM virtual machine instead.
\footnote{
Our first attempts to implement 
the user task in Solidity (Ethereum blockchain) turned out to 
be highly cumbersome activity due to limited expressive
power of the language and limitations of a virtual machine. 
 }

As for the user task, we took UNSAT problem: 
the problem of determining that there is no 
satisfying assignment for a Boolean formula. 
UNSAT problem is conjugated to a famous NP-complete
problem SAT 
\cite{cook1971complexity},
but unlike SAT, its solution can not be checked in polynomial
time in general case.
Our choice is justified by the fact that an algorithm for
solving UNSAT problem has a direct practical application
 - a
symbolic model-checking of computer algorithms. 

There are many algorithms for solving SAT/UNSAT
problem, and it is known that no one of them 
guarantees a working
time better, than $O(2^n)$, where $n$ - the number of
boolean variables in Boolean formula. Still, they are 
quite good in practice.

We have chosen to implement one of the simplest such 
algorithms - DPLL \cite{davis1962machine}. 
There are many DPLL implementations available, but in 
our case, we had to implement it in an iterative form 
discussed in Section \ref{fixpoint}.

As for the input, passed into DPLL algorithm, we 
took an program equivalence checking problem. 
Particulary, we check equivalence of programs where 
each program implement a FIFO queue, but in different
way (the problem is taken from \cite{biere99}). 

The problem is encoded in a 
Conjunctive Normal Form (CNF)formula $d$, 
such that if $DPLL(d) = Unsat$, 
it means an equivalence of programs. 
Otherwise the algorithm outputs a counter-example.

In this experiment, we were interested in the following:
\begin{enumerate}
\item 
How much the computation time of the function $f$ will
change between the usual recursive form ($T_1$) and 
its iterative form ($T_2$).
\item 
How much a bytecode size of the function $f$ will change
between the usual recursive ($S_1$) and iterative ($S_2$)
forms.

\item 
What will be a number of iterations needed to get 
from the initial point to a fixpoint of $f$ (denoted as $n$);
what will be a certificate size ($C_f$)?

\item 
What will be the greatest size of an intermediary function
result $f^i(d)$ (denoted as $d_{max}$), and the size
of initial point $d_0$? 
\end{enumerate}

Experiments were carried out on a computing node
Intel Xeon Gold 6254 CPU 3.10GHz x 4 Cores, 16GB RAM;
Execution environment: Erlang 20, ERTS 9.3.3.11, running
under Linux Fedora 29 OS.

\begin{center}
\begin{tabular}{|c|c|c|c|c|c|c|c|c|c|c|}
\hline
	 \ & $d_0$ & $S_1$ & $S_2$ & $T_1$ & $T_2$ & 
	 $d_{max}$ & $C_f$ & $n$ \\
\hline
	$QueueInvar_2$ & 5703 & 2064 & 2272 & 0.1 & 0.4 & 42028 & 52992 & 
	1656 \\
\hline
	$QueueInvar_4$ & 13707 & 2064 & 2272 & 193 & 423 & 170758 & 13145664 & 410802 \\
\hline
\end{tabular}
\end{center}
Here $T_1$, $T_2$ is measured in seconds, 
$d_0, d_{max}, S_1, S_2, C_f$ - in bytes,
$n$ - number of iterations.
\bigbreak
Besides, we implemented the solution 
refutation step scenario consisting of the 
following steps:
\begin{enumerate}
\item The user task DPLL(x) presented in an
iterative form is published together with
the initial data $d$.

The size of a user task $S_2$ and its data $d_0$ permits to send it into TWCD - the
smart-contract.

\item 
Some participant sends intentionally incorrect
solution, consisting of values
$r_n, cp, hc$.
It might be the case that the size of $cp$ 
(measured as a fraction of $C_f$) is
larger than smart-contract is able to process.
In this case, $cp$ is published on an external
immutable storage, for example IPFS.
Instead of $cp$, it sends a link to the 
data object residing on IPFS.
It is a responsibility of the solution
provider to make this link alive and available. 

\item 
Some other participant, after checking
the published solution and detecting the
error, constructs a refutation. The 
refutation $i$, $c_{i-1}$, $c_i$, $r_{i-1}$ is sent into the smart-contract.

\item 
The smart-contract checks the provided
refutation by applying rules
\ref{refute_check1}, \ref{refute_check2}, \ref{refute_check3} from Section 
 \ref{solution}. 
If the link to IPFS is published instead of
pure $cp$ value, then the smart-contract asks
a trusted oracle to provide data locating 
at $i-1$ and $i$ offset, receiving $c_{i-1}$,
$c_i$. 
If the IPFS link is unavailable, the solution
is declined.
\end{enumerate}

A Proof-of-Concept implementation of the
protocol together with described scenario
and the user task program 
is available at the repository \cite{sources}.

All IPFS and Oracle-related functionality is modelled using a usual program code, without
making any external service calls.   

\section{Related Works} \label{related}
Certification of computation integrity performed by
some party we do not fully trust is considered a 
classic problem in modern cryptography.
Several solutions to this problem have been proposed,
each solution with its own set of compromises.
One set of methods rely on a so-called
 \textit{Probabilistically Checkable Proofs (PCP)}
theorem for computations in class NP.
  
A reminder: NP class consists of decision problems 
such that there exists a verification procedure able 
to check its solutions in polynomial time.
 
PCP theorem states that any solution for an NP problem
could be checked by a polynomial time algorithm using a 
fixed number of randomly chosen structural elements
of a solution \cite{pcp}.

Methods belonging to this family could be 
characterized by the following pattern of interaction:
 1) A user defines his computation task in a form of a 
 boolean circuit
 (or functional elements circuit \cite{lozhkin});
 constructs a special polynomial - 
 Algebraic Normal Form (ANF; amount of terms 
 in ANF depends on the number of functional elements 
 in the circuit \cite{vollmer2013introduction});
 sends the circuit and the input data to 
 the provider
 2) The provider derives corresponding boolean
 function; using this function, he gets the ANF 
 and encodes the output of each of the circuit nodes 
 in ANF; computes both the result of a computation 
 and a certificate. The certificate can be stored 
 on the provider's part, in the cloud or transmitted
 with the result to the user
 3) To verity the certificate, the user utilizes one 
 of approaches shown below. Generally, in order to
 verify a certificate, it is enough to randomly 
 chose a limited number of circuit nodes for 
 each approach, as PCP theorem implies.
 
 \begin{figure}[t]
 \captionsetup{width=.8\linewidth}
\begin{center}
\begin{tabular}{|c|c|c|c|c|}
\hline
\, & $inter(n)$ & $proof(n)$ & $verify(n)$ & $cert(n)$ \\
\hline
Libra Non-ZKP \footnotemark & $O(d(n) \cdot \log{s(n)})$ & $O(s(n))$ & $O(d(n) \cdot \log s(n))$ & $O(d(n) \cdot \log s(n))$ \\
\hline
TrueBit & $O(\log n)$ & 
 $ O(n \log{n})$ & $O(1)$ & $O(n)$ \\
\Xhline{3\arrayrulewidth}
SafeComp & 1 & $O(n)$ & $O(1)$ & $O(n)$ \\
\Xhline{3\arrayrulewidth}

\end{tabular}
\end{center}
\caption{
\footnotesize{
	Asymptotic estimates for several certification algorithms.
	Denotation:
		$proof(n)$ - complexity of proof construction,
		$verify(n)$ - complexity of proof verification,
		$inter(n)$ - number of interaction rounds,
		$cert(n)$ - size of certificate,\\
		$s(n)$ - size of functional element circuit as 
		a function of task input size,
		$d(n)$ - depth of a circuit, i.e. a maximum path length in the circuit.
}}
\label{fig:diff}
\end{figure}
 \footnotetext{Libra is a Zero-Knowledge 
 (ZK) protocol: it hides the user 
 function input from the prover, and 
 this is not what we would like to compare
 with, because, both SafeComp and TrueBit do not hide user input. 
 If we omit ZK property,
  Libra inherits communication complexity
  of GKP protocol \cite{goldwasser2015delegating} with some
  modifications. }
 There are known approaches to certificate 
 verification:
 1) Interactive approach without commitments: 
 Query a certificate (residing on a provider's part
 or in a cloud) at a different nodes until the user is
 satisfied \cite{goldwasser2015delegating}
 \cite{xie2019libra} \cite{thaler2013time}
2) Interactive approach with commitments: 
form commitments with cryptographic protocols 
that verifies the integrity of a certificate. 
The key difference from the previous approach is 
a significant extension of the class of 
 verification tasks
 \cite{setty2013resolving} \cite{setty2012taking}
 \cite{parno2013pinocchio}
 	   
PCP-inspired approaches impose considerable limitations
on verifiable computations. For instance, accurate
upper estimates
for the number of loop iterations and the size of data
structure ought to be known in advance.

Protocol Libra
\cite{xie2019libra}
belongs to those methods relying on PCP theorem, and
is one of the latest methods known from the literature.
It cannot be directly matched with our protocol because it has a number of limitations, for example, on the structure of user's computations. Nevertheless, in some cases it can be applied to obtain the same technical result. Therefore, we also give bounds for it.

In the table \ref{fig:diff}, we compare the nearest known analogues in terms of their asymptotic behaviour.
 
\subsubsection*{Comparision with TrueBit}
With the advent of blockchains equipped with smart-contracts programming capability, another approach for solving 
the computation certification problem has emerged.
In this setting, a smart-contract is used as 
a transparent autonomous arbiter able 
to resolve computation disputes among parties.

The first protocol to explore this idea was
TrueBit \cite{truebit}. We now give a brief
compari\-son,
highlighting the main differences between
TrueBit and SafeComp.
\\
\textbf{Solution refutation procedure.}
If a dispute against a solution arises, 
TrueBit entails an interactive game between two
 computation providers with the aim of finding
 the very first incorrect computation step
 and convincing the smart-contract that this
 step is indeed wrong. 
The game is played in several rounds.
The number
of rounds is bounded by $O(\log{n})$, where $n$
is the number of computation steps.
On each round, a solver and another provider compute
several Merkle tree roots: 
the computation states are taken for the leafs.
Those values are needed to find the source of
disagreement.
Complexity of computing one such root is bounded by
$O(n)$.
Total computation complexity for constructing a
refutation, for all rounds, is bounded by $O(n \log{n})$.
The advantage of such approach is that it allows
parties to optimize the data volume sent into a
smart-contract, but with an extra cost of computing
Merkle Tree roots, several times.

In SafeComp, we use another approach: a solver
publishes both a solution to a task and a certificate
projection - a partially disclosed certificate 
sequence (Section \ref{solution_simple}, p.\ref{cert_proj}) - that allows other parties to find a
divergent part of a certificate fast, in $O(\log{n})$ steps using binary search. 
Both the size and computation complexity for a certificate has a bound $O(n)$.
If a certificate projection is large, we use an
external immutable  storage, IPFS for example,
using an oracle technology to move data into a smart-contract in case it is needed.
The solution refutation procedure in our case is, thereby, done in 1 round, lowering the complexity
of the procedure from $O(n \log{n})$ to $O(n + \log{n}) = O(n)$, but at the expense of 
writing extra $O(n)$ data, i.e. a certificate projection, into the network.  
\\
\textbf{Form of Computational Task.}
In case of TrueBit, a user task must be compiled
into a bytecode of some virtual machine. One have
to place an interpreter for the bytecode in the
blockchain, so a smart-contract will be able to
resolve disputes.  
In our case, one needs to present its computation
in a special - iterative - form, meaning that it may
entail program rewriting in a different style. 
This may cause inconvenience, but maybe soften by
the fact that one does not need to place 
an interpreter on the blockchain. 
We are unsure if a suitable interpreter exists at the moment. 
 
\section{Conclusion and Future Work}
In this article, we presented a protocol called 
SafeComp aiming to solve
a cloud computation integrity certification problem 
in a specified context. 

We have given a formal specification of the protocol,
proved some of its security properties within 
a specified threat model.  

Comparing to the nearest related work, SafeComp lowers
the refutation procedure complexity considerably,
and communication complexity becomes exactly one round.

We evaluated the protocol and made some measures that
convinces us in viability of the protocol in practice. 
For the future work, we would like to do the following:

1) 
Propose a provably reliable economic
incentives model for participants. 
As for now, we almost completely omitted the subject. 

2) 
Implement SafeComp protocol within some 
public blockchain, and perform a computation
certification for some practically interesting computations.

3) Investigate a relation between an 
iterative representation of a user program and 
its intermediary state size, comparing to the 
virtual machine byte-code representation. 

\subsubsection*{Acknowledgements}
We would like to thank Andrey Rybkin and Mikhail Borodin
for their constructive critics of the protocol: during
one of discussions a vulnerability in the protocol logic 
was found and fixed. We also thank Jason
Teutsch for his thoughtful comments and
suggestions on this work, together with
clarification on TrueBit internals. 

\def\refname{References}
\bibliographystyle{plain}
\bibliography{safecomp}

\section*{Appendix 1} \label{appendix:proofs}
\begin{theorem-non}
For all functions $C: \mathbb{N} \rightarrow \mathbb{N}$ computable at some point $d$ there exist computable functions
$$inj: \mathbb{N} \rightarrow T, 
proj : T \rightarrow \mathbb{N}, F: T \rightarrow T$$ such that
$$ proj (\textbf{fix} \,\, F \, inj(d)) = C(d) $$
where $\textbf{fix} : (T \rightarrow T) \times \mathbb{N} \rightarrow T$ - an operator that calculates fixed point of a function starting from a given point. That is,
$$ \textbf{fix}\,F \, p = F(\textbf{fix} \, F \, p) $$
$T$ - a finite set.
\end{theorem-non}
\begin{proof}
Let
$$inj(d) = \langle d, 0 \rangle \text{ \,\, \,\,      } 
proj( \langle x, y \rangle ) = x$$
$$F(\langle x, 0 \rangle) = \langle C(x), 1 \rangle \text{ \,\, \,\,      }  F(\langle x, 1 \rangle) = \langle x, 1 \rangle$$
Selected functions satisfy conditions of the theorem. Such representation of $F$
could be called \textit{trivial}: it does not help to divide function $C(x)$ into composite elementary units.
But is it possible to represent $F$ in any non-trivial way? Yes, we can.
The following argument proofs this fact in a  constructive way. 

According to Turing thesis, every
function $C(x)$, computable at a point, can be represented as a Turing machine. Let $M$
be a machine corresponding to the function $C(x)$, that is
$$ M = (Q, \Sigma, \Gamma, \delta, q_0, B, Q_f) $$
$$ Q - \text{non-empty set of states}, \Gamma - \text{non-empty set of tape alphabet symbols}$$
$$ \Sigma \subseteq \Gamma - \text{set of input symbols}, Q_f \subseteq Q - \text{set of possible states}$$
$$ q_0 \in Q - \text{initial state} $$
$$ \delta: Q \times \Gamma \rightarrow Q \times \Gamma \times \{L,R\}  - \text{transition function}$$
Lets define the function $F$ as:
\begin{gather*}
F( \langle M, I, q, p  \rangle) = 
\begin{cases} 
\langle M, I', q', p' \rangle, \text{ if } q \notin Q_f\\
\langle M, I, q, p  \rangle, \text{ otherwise }
\end{cases}
\end{gather*}
where
$$ (q', X, p') = \delta(q, I_p), \text{ where } I_p \text{ is a content of the tape's cell at the position } p $$
$$ I' = I[X / p] $$
Starting from the position $p$ and the state $q$, the machine $M$ changes state of the tape  from $I$ to $I'$ (replacing content of $p$ with $X$) and the position of its head $p'$. If it reaches an accepting state (i.e. a member of $Q_f$)
it wouldn't do any steps further, so 
a fixpoint of $F$ is obtained.

Theorem statement suggests that a function $C$ is computable at a point $d$, so, after some number of iterations, the machine will necessarily reach an accepting state.

Notice that elements $\langle M, I, q, p \rangle$ can be encoded, for instance, by a finite number of
natural numbers. An example of such an encoding can be found in \cite{hopcroft}.

The function $F$ is computable because it  relies on a few computable functions: 
a transition function 
$\delta(q,X)$, a function of checking whether an element belongs to a set, i.e $q \in Q_f$, and a choice function $if-then-else$. All these functions are computable if sets $Q$ and 
$\Gamma$ are finite.

Let $inj(d) := \langle M, I_0, q_0, 1 \rangle$, where $I_0$ is an initial state of the tape that contains $d$ among other things.
Let $proj(\langle M, I, q, p \rangle)$ be a function that extracts the value $C(d)$ from the tape $I$.
Such functions could be constructed because it
is always possible to allocate space for an initial
 value and for an answer.
The answer $C(d)$ is presented as a finite number
of cells on the tape since computability of the
function at the point guarantees a finite number 
of steps of the machine.

If a size of an answer can not be known in advance,
we could always fix the initial cell and introduce a
unique mark denoting an end of area reserved for 
an answer.

The described representation of functions $F, inj, proj$ satisfies requirements of the theorem and is not trivial.
\end{proof}
\end{document}